\documentclass[11pt]{article}%
\usepackage{amsmath}
\usepackage{amsfonts}
\usepackage{amssymb}
\usepackage{graphicx}
\usepackage{fullpage}%
\providecommand{\U}[1]{\protect\rule{.1in}{.1in}}
\newtheorem{theorem}{Theorem}

\newtheorem{claim}[theorem]{Claim}

\newtheorem{definition}[theorem]{Definition}

\newtheorem{lemma}[theorem]{Lemma}

\newtheorem{proposition}[theorem]{Proposition}

\newenvironment{proof}[1][Proof]{\noindent\textbf{#1.} }{\ \rule{0.5em}{0.5em}}
\begin{document}

\title{On Perfect Completeness for QMA}
\author{Scott Aaronson\thanks{Email: aaronson@csail.mit.edu. \ Supported by MIT\ CSAIL
and the Keck Foundation.}\\MIT}
\date{}
\maketitle

\begin{abstract}
Whether the class\ $\mathsf{QMA}$\ (Quantum Merlin Arthur)\ is equal to
$\mathsf{QMA}_{1}$, or $\mathsf{QMA}$\ with one-sided error,\ has been an open
problem for years. \ This note helps to explain why the problem is
difficult,\ by using ideas from real analysis to give a \textquotedblleft
quantum oracle\textquotedblright\ relative to which $\mathsf{QMA}%
\neq\mathsf{QMA}_{1}$. \ As a byproduct, we find that there are facts about
quantum complexity classes that are classically relativizing but not quantumly
relativizing, among them such \textquotedblleft trivial\textquotedblright%
\ containments as $\mathsf{BQP}\subseteq\mathsf{ZQEXP}$.

\end{abstract}

\section{Introduction\label{INTRO}}

The complexity class $\mathsf{MA}$\ (Merlin-Arthur) was introduced by Babai
\cite{babai:am2}\ in\ 1985. \ Intuitively, $\mathsf{MA}$\ is a probabilistic
version of $\mathsf{NP}$; it contains all problems for which an omniscient
wizard Merlin can convince a probabilistic polynomial-time verifier Arthur of
a \textquotedblleft yes\textquotedblright\ answer, by a one-round protocol in
which Merlin sends Arthur a purported proof $z$, and then Arthur checks $z$.
\ In the usual definition, if the answer to the problem is \textquotedblleft
yes\textquotedblright\ then there should exist a string $z$ that makes Arthur
accept with probability at least $2/3$ (this property is called
\textit{completeness}), while if the answer is \textquotedblleft
no\textquotedblright\ then no $z$ should make Arthur accept with probability
more than $1/3$ (this property is called \textit{soundness}).

One of the first questions people asked about $\mathsf{MA}$\ was whether it
can be made to have \textit{perfect completeness} (also called
\textit{one-sided error}): that is, whether the $2/3$ in the above definition
can be replaced by $1$. \ In other words, can we assume without loss of
generality that \textit{Arthur never rejects a valid proof?} \ This question
was answered in the affirmative by Zachos and F\"{u}rer \cite{zf}, using a
technique introduced earlier by Lautemann \cite{lautemann}\ to show
$\mathsf{BPP}\subseteq\mathsf{\Sigma}_{\mathsf{2}}^{\mathsf{P}}$\ (for a
different proof see Goldreich and Zuckerman \cite{gz}).

A decade ago, Kitaev \cite{ksv} and Watrous \cite{watrous} introduced a
quantum analogue of $\mathsf{MA}$, called $\mathsf{QMA}$\ (Quantum Merlin
Arthur). \ Loosely speaking, $\mathsf{QMA}$\ is the same as $\mathsf{MA}$,
except that the verifier Arthur is a polynomial-time \textit{quantum}
algorithm, and the proof sent by Merlin is a quantum state $\left\vert
\psi\right\rangle $\ with polynomially many qubits. \ We know a reasonable
amount about $\mathsf{QMA}$ (see Aharonov and Naveh \cite{an}\ for a survey).
\ Like $\mathsf{MA}$, for example, $\mathsf{QMA}$ allows exponential
amplification of completeness and soundness \cite{mw}, is contained in
$\mathsf{PP}$\ \cite{mw}, and has natural complete promise problems \cite{ksv}.

However, the question of whether $\mathsf{QMA}$\ can be made to have perfect
completeness has resisted attack. \ At first a mere nuisance, this question
has increasingly cropped up in quantum complexity theory. \ For example, two
years ago Bravyi \cite{bravyi}\ defined a quantum analogue of the $k$-SAT
problem, and showed it complete for the complexity class $\mathsf{QMA}_{1}$,
meaning $\mathsf{QMA}$\ with one-sided error. \ But showing quantum
$k$-SAT\ is $\mathsf{QMA}$-complete would require further showing that
$\mathsf{QMA}_{1}=\mathsf{QMA}$, or equivalently, that $\mathsf{QMA}%
$\ protocols can be made to have perfect completeness. \ What makes the
situation even stranger is that, if we allow multiple rounds of interaction
between the prover and verifier (yielding the class $\mathsf{QIP}$), then
quantum interactive proof systems \textit{can} be made to have perfect
completeness \cite{kitaevwatrous}.

In this note we help explain this puzzling state of affairs, by giving a
\textit{quantum oracle}\ $\mathcal{U}$ relative to which $\mathsf{QMA}%
_{1}^{\mathcal{U}}\neq\mathsf{QMA}^{\mathcal{U}}$. \ A quantum oracle, as
defined by Aaronson and Kuperberg \cite{ak}, is simply a unitary
transformation on quantum states that can be applied in black-box fashion.
\ Our result implies that there is no \textquotedblleft black-box
method\textquotedblright\ to convert $\mathsf{QMA}$\ protocols into
$\mathsf{QMA}_{1}$ protocols, in the same sense that there are black-box
methods to convert $\mathsf{MA}$ protocols into $\mathsf{MA}_{1}$
protocols,\ and to convert $\mathsf{QMA}$\ protocols\ into $\mathsf{QMA}%
$\ protocols with exponentially small error. \ If a proof of $\mathsf{QMA}%
_{1}=\mathsf{QMA}$ exists, it will instead have to use \textquotedblleft
quantumly nonrelativizing techniques\textquotedblright: techniques that are
sensitive to the presence of quantum oracles.

Somewhat surprisingly, our separation proof has almost nothing to do with
complexity theory, and instead hinges on real analysis. \ Our oracle will act
on just a single qubit, and will rotate that qubit by an angle $\theta$\ that
is either $0$ or far from zero. \ We will show that any $\mathsf{QMA}_{1}%
$\ protocol to convince a time-bounded verifier that $\theta\neq0$, using any
finite-sized quantum proof, would lead to a matrix $E\left(  \theta\right)
$\ that depends analytically on $\theta$, yet whose maximum eigenvalue has the
\textquotedblleft piecewise\textquotedblright\ behavior shown in Figure
\ref{eigenceil}. \ We will then use results from real analysis to show that
this behavior cannot occur.%

\bigskip

\bigskip

\begin{center}
\includegraphics[
natheight=2.008100in,
natwidth=3.011300in,
height=2.0081in,
width=3.0113in
]%
{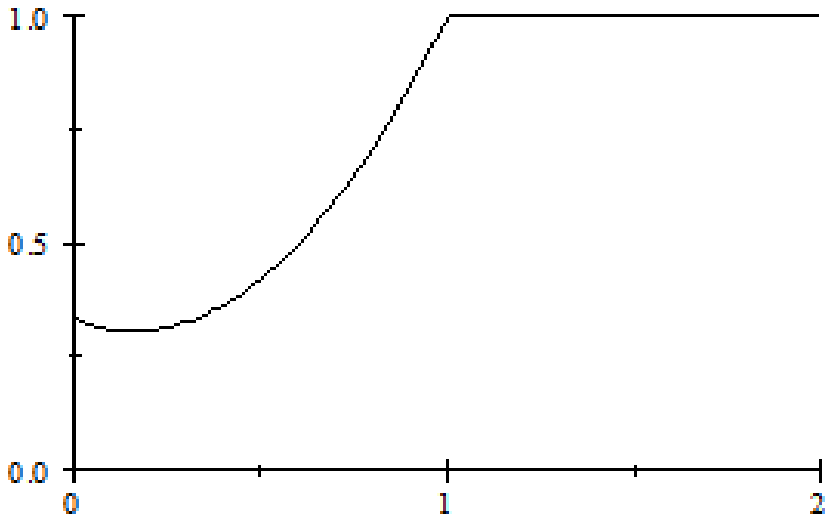}%
\\
Figure 1: As we vary $\theta$, the largest eigenvalue $a\left(  \theta\right)
$ of the matrix $E\left(  \theta\right)  $\ must start out small, but then
\textquotedblleft plateau\textquotedblright\ at $a\left(  \theta\right)  =1$.
\ We will show that this contradicts the analyticity of $E\left(
\theta\right)  $.
\label{eigenceil}%
\end{center}

Since our argument does not depend on the running time of the $\mathsf{QMA}%
_{1}$ machine (so long as it is finite), the same argument will yield quantum
oracles $\mathcal{U}$\ such that $\mathsf{BQP}^{\mathcal{U}}\not \subset
\mathsf{QMA}_{1}\mathsf{EXP}^{\mathcal{U}}$ and even $\mathsf{BQP}%
^{\mathcal{U}}\not \subset \mathsf{QMA}_{1}\mathsf{EEEXP}^{\mathcal{U}}$.
\ This, in turn, has a somewhat surprising implication: that quantum oracles
can invalidate even complexity class containments that hold for
\textquotedblleft trivial\textquotedblright\ reasons in the unrelativized
world. \ In particular, we will argue that there are extremely simple proof
techniques---including the representation of quantum amplitudes by explicit
sequences of bits---that are classically relativizing but not quantumly
relativizing (at least when applied to one-sided-error complexity classes).
\ Unfortunately, knowing this does not by itself seem to help in finding a
proof that $\mathsf{QMA}_{1}=\mathsf{QMA}$.

Some argue that any quantum complexity class involving perfect completeness is
\textquotedblleft inherently unphysical,\textquotedblright\ and we do not wish
to dispute this. \ Indeed, our results could even be taken as further evidence
for this point of view. \ On the other hand, classes like $\mathsf{QIP}$\ and
$\mathsf{QMA}$\ could also be seen as \textquotedblleft
unphysical\textquotedblright\ (since there are no Merlins), yet few quantum
computing researchers would deny that their study has led to major insights.
\ On a related topic, let us stress that our result does not depend on
restricting the set of gates available to the $\mathsf{QMA}_{1}$\ machine: it
works assuming \textit{any} countable set of gates. \ The key issue, then, is
not any limitation of the gate basis, but simply the underlying requirement of
perfect completeness.

The rest of the paper is organized as follows. \ Section \ref{PRELIM}\ reviews
some preliminaries from complexity theory and real analysis, Section
\ref{RESULT} proves the main result, Section \ref{DISC}\ discusses the
implications for quantum oracles, and Section \ref{EXT}\ concludes with some
extensions and open problems.

\section{Preliminaries\label{PRELIM}}

In what follows, we assume familiarity with standard complexity classes such
as $\mathsf{QMA}$\ and $\mathsf{MA}$. \ See the Complexity
Zoo\footnote{http://www.complexityzoo.com} for definitions. \ For
completeness, we now define the class $\mathsf{QMA}_{1}$, or $\mathsf{QMA}%
$\ with one-sided error.

\begin{definition}
A language $L\subseteq\left\{  0,1\right\}  ^{\ast}$\ is in $\mathsf{QMA}_{1}$
if\ there exists a uniform polynomial-size quantum circuit family $\left\{
C_{n}\right\}  _{n\geq1}$, and a polynomial $p$, such that for all inputs
$x\in\left\{  0,1\right\}  ^{n}$:

\begin{itemize}
\item (Perfect Completeness) If $x\in L$, then there exists a $p\left(
n\right)  $-qubit quantum witness $\left\vert \varphi\right\rangle $ such that
$C_{n}\left(  x,\left\vert \varphi\right\rangle \right)  $\ accepts with certainty.

\item (Constant Soundness) If $x\notin L$, then $C_{n}\left(  x,\left\vert
\varphi\right\rangle \right)  $\ accepts with probability at most $1/2$\ for
all $\left\vert \varphi\right\rangle $.
\end{itemize}
\end{definition}

One can similarly define $\mathsf{MA}_{1}$ as the class of languages $L$\ for
which there exists a randomized polynomial-time algorithm $A$ such that for
all inputs $x\in\left\{  0,1\right\}  ^{n}$, if $x\in L$\ then there exists a
witness $w\in\left\{  0,1\right\}  ^{p\left(  n\right)  }$\ such that
$A\left(  x,w\right)  $\ accepts with certainty, while if $x\notin L$\ then
$A\left(  x,w\right)  $\ accepts with probability at most $1/2$\ for all $w$.
\ As mentioned before, the result of Zachos and F\"{u}rer \cite{zf} implies
that $\mathsf{MA}_{1}=\mathsf{MA}$, whereas we do not know whether
$\mathsf{QMA}_{1}=\mathsf{QMA}$.

Note that, because of the perfect completeness condition, the definition of
$\mathsf{QMA}_{1}$\ might depend on the particular basis of gates used to
generate $C_{n}$. \ Indeed, the natural way to show that $\mathsf{QMA}_{1}%
$\ does \textit{not} depend of the basis of gates would presumably be to show
that $\mathsf{QMA}_{1}=\mathsf{QMA}$, the very task for which we are pointing
out an obstacle! \ For our purposes, though, we can take \textit{any}
countable set of 1- and 2-qubit gates as the gate basis of the $\mathsf{QMA}%
_{1}$\ machine, regardless of how many bits are needed to describe those
gates. \ For example, we could take the set of all 1- and 2-qubit gates that
are computably describable. \ Our separation results will still go through,
and such an assumption can only make our results stronger. \ (Furthermore, the
only reason the gate basis needs to be countable is so that a diagonalization
argument will go through. \ If the quantum oracle $U$ could be chosen
\textit{subsequent} to the choice of $\mathsf{QMA}_{1}$\ machine, then we
could even handle 1- and 2-qubit gates with arbitrary complex-valued
transition probabilities.)

Following Aaronson and Kuperberg \cite{ak}, we define a quantum oracle
$\mathcal{U}$\ to be simply a collection of unitary operations $\left\{
U_{n}\right\}  _{n\geq1}$, where each $U_{n}$\ acts on some number of qubits
$q\left(  n\right)  $\ (in this note $q\left(  n\right)  $ will always be
$1$). \ Let $\mathsf{C}$\ be a quantum complexity class. \ Then by
$\mathsf{C}^{\mathcal{U}}$, we mean the class of problems solvable by a
$\mathsf{C}$\ machine that can, at any time step, apply any $U_{n}%
\in\mathcal{U}$\ to any subset of its qubits at unit cost. \ While this is
admittedly an informal definition, for any $\mathsf{C}$ of interest to us it
is easy to give a reasonable formalization. \ While there \textit{are}
ambiguities in defining $\mathsf{C}^{\mathcal{U}}$, none of those ambiguities
will turn out to matter for us. \ For example, we can assume (if we like) that
a $\mathsf{C}^{\mathcal{U}}$ machine\ is also able to apply $U_{n}^{-1}$\ and
controlled-$U_{n}$, possibly with different values of $n$ in different
branches of a superposition. \ None of these decisions will affect our results.

We now turn to reviewing some facts from real analysis. \ Recall that a
function $f:\mathbb{R}\rightarrow\mathbb{R}$\ is called \textit{real analytic}
if for every $x_{0}\in\mathbb{R}$, the Taylor series about $x_{0}$ is
convergent and equal to $f\left(  x\right)  $\ for all $x$ close enough to
$x_{0}$. \ Every real analytic function is smooth, but the converse does not hold.

We will need the following theorem of Alekseevsky et al.\ (Theorem 5.1 in
\cite{aklm}):

\begin{theorem}
[\cite{aklm}]\label{aklmthm}Let%
\[
p\left(  \theta\right)  \left(  x\right)  =b_{0}\left(  \theta\right)
+b_{1}\left(  \theta\right)  x+b_{2}\left(  \theta\right)  x^{2}+\cdots
+b_{N}\left(  \theta\right)  x^{N}%
\]
be a real polynomial in $x$ with all real roots, parameterized by $\theta
\in\mathbb{R}$. \ Suppose the coefficients $b_{0}\left(  \theta\right)
,\ldots,b_{N}\left(  \theta\right)  $\ are all real analytic functions of
$\theta$. \ Then there exist real analytic functions $\lambda_{1}\left(
\theta\right)  ,\ldots,\lambda_{N}\left(  \theta\right)  $\ such that
$\left\{  \lambda_{1}\left(  \theta\right)  ,\ldots,\lambda_{N}\left(
\theta\right)  \right\}  $\ is the set of roots of $p\left(  \theta\right)
\left(  x\right)  $\ for all $\theta\in\mathbb{R}$.
\end{theorem}

We also need the following basic fact:

\begin{proposition}
\label{analprop}Let $f:\mathbb{R}\rightarrow\mathbb{R}$\ be a real analytic
function. \ If there exists an open set $\left(  x,y\right)  \subset
\mathbb{R}$\ on which $f$ is constant, then $f$\ is constant everywhere.
\end{proposition}

Note that Proposition \ref{analprop}\ is false with smooth functions in place
of real analytic ones.\footnote{The standard counterexample is $f\left(
x\right)  =0$\ for $x\leq0$ and $f\left(  x\right)  =e^{-1/x^{2}}$\ for
$x>0$.} \ This is why we need analyticity for our result.

\section{Result\label{RESULT}}

We first need a more-or-less standard fact (proved for completeness), which
recasts the problem of finding an optimal $\mathsf{QMA}$\ witness as a
principal eigenvalue problem.

\begin{lemma}
\label{vlem}Let $V$ be a quantum verifier that takes as input a $Q$-qubit
quantum witness $\left\vert \varphi\right\rangle $, and that makes
$T$\ queries to a quantum oracle described by a unitary matrix $U$. \ Also,
let $a\left(  U\right)  $\ be the acceptance probability of $V^{U}$ maximized
over all possible $\left\vert \varphi\right\rangle $. \ Then there exists a
$2^{Q}\times2^{Q}$\ complex-valued matrix $E\left(  U\right)  $ such that

\begin{enumerate}
\item[(i)] Every entry of $E\left(  U\right)  $\ is a polynomial in the
entries of $U$, of degree at most $2T$.

\item[(ii)] $E\left(  U\right)  $ is Hermitian for all $U$.

\item[(iii)] $a\left(  U\right)  $ equals the largest eigenvalue of $E\left(
U\right)  $, for all $U$.
\end{enumerate}
\end{lemma}

\begin{proof}
Let $a\left(  U,\left\vert \varphi\right\rangle \right)  $\ be the acceptance
probability of $V^{U}$\ on input $\left\vert \varphi\right\rangle $. \ Then
clearly there exist vectors $\left\{  \left\vert v_{i}\right\rangle \right\}
_{i=1}^{2^{Q}}$\ (not necessarily normalized, and depending on $U$) such that%
\[
a\left(  U,\left\vert \varphi\right\rangle \right)  =\sum_{i}\left\vert
\left\langle v_{i}|\varphi\right\rangle \right\vert ^{2}.
\]
Furthermore, by an observation of Beals et al.\ \cite{bbcmw}, every entry of
every $\left\vert v_{i}\right\rangle $\ must be a polynomial in the entries of
$U$, of degree at most $T$. \ (This is because initially the entries are
degree-$0$ polynomials, and every query to the oracle can increase the degree
by at most $1$.) \ So if we set $E:=\sum_{i}\left\vert v_{i}\right\rangle
\left\langle v_{i}\right\vert $, then $E$ is a $2^{Q}\times2^{Q}$\ Hermitian
matrix, every entry of which is a polynomial of degree at most $2T$.
\ Furthermore $a\left(  U,\left\vert \varphi\right\rangle \right)
=\left\langle \varphi|E|\varphi\right\rangle $, which implies that%
\[
a\left(  U\right)  =\max_{\left\vert \varphi\right\rangle }a\left(
U,\left\vert \varphi\right\rangle \right)  =\max_{\left\vert \varphi
\right\rangle }\left\langle \varphi|E|\varphi\right\rangle
\]
which is just the largest eigenvalue of $E$.
\end{proof}

We now prove the main result.

\begin{theorem}
\label{osep}There exists a quantum oracle $\mathcal{U}$ such that
$\mathsf{QMA}_{1}^{\mathcal{U}}\neq\mathsf{QMA}^{\mathcal{U}}$.
\end{theorem}

\begin{proof}
Let $\theta$\ be a real number, and let $U=U\left(  \theta\right)  $\ be the
one-qubit unitary transformation%
\[
\left(
\begin{array}
[c]{cc}%
\cos\theta & -\sin\theta\\
\sin\theta & \cos\theta
\end{array}
\right)  .
\]
Given oracle access to $U$, we consider the problem of deciding whether
$\theta=0$ (the NO case)\ or $1\leq\theta\leq2$\ (the YES case), promised that
one of these holds. \ Of course this problem is easily solved by a quantum
computer, with bounded error probability, using $O\left(  1\right)  $\ queries
to $U$. \ On the other hand, we will show that this problem does not admit a
perfect-completeness $\mathsf{QMA}$\ protocol, with any finite number of
queries to $U$ and any finite-sized quantum proof.

To see this, let $V$ be a verifier, let $T$\ be the number of queries that $V$
makes to $U$,\ and let $Q$\ be the number of qubits in $V$'s quantum witness.
\ Also, let $a\left(  \theta\right)  $\ be the acceptance probability of $V$
assuming $U=U\left(  \theta\right)  $, maximized over all $Q$-qubit quantum
witnesses $\left\vert \varphi\right\rangle $. \ Then by Lemma \ref{vlem},
there exists a $2^{Q}\times2^{Q}$\ complex-valued matrix $E\left(
\theta\right)  $\ such that

\begin{enumerate}
\item[(i)] Every entry of $E\left(  \theta\right)  $\ is a polynomial in
$\cos\theta$ and $\sin\theta$, of degree at most $2T$.

\item[(ii)] $E\left(  \theta\right)  $ is Hermitian for all $\theta
\in\mathbb{R}$.

\item[(iii)] $a\left(  \theta\right)  $ equals the largest eigenvalue of
$E\left(  \theta\right)  $, for all $\theta\in\mathbb{R}$.
\end{enumerate}

Let $N=2^{Q}$, and let $\lambda_{1}\left(  \theta\right)  ,\ldots,\lambda
_{N}\left(  \theta\right)  $\ be the eigenvalues of $E\left(  \theta\right)
$. \ Then the $\lambda_{i}\left(  \theta\right)  $'s are roots of a degree-$N$
characteristic polynomial parameterized by $\theta$:%
\[
p\left(  \theta\right)  \left(  x\right)  =b_{0}\left(  \theta\right)
+b_{1}\left(  \theta\right)  x+b_{2}\left(  \theta\right)  x^{2}+\cdots
+b_{N}\left(  \theta\right)  x^{N}.
\]
Each coefficient $b_{j}\left(  \theta\right)  $\ is a polynomial in the
entries of $E\left(  \theta\right)  $\ of degree at most $N$, and hence, by
(i), a polynomial in $\cos\theta$ and $\sin\theta$\ of degree at most $2TN$.
\ By (ii), the $\lambda_{i}\left(  \theta\right)  $'s\ are all real, and
therefore the $b_{j}\left(  \theta\right)  $'s must be real as well for all
$\theta$. \ Combining these facts, we find that each $b_{j}\left(
\theta\right)  $ is a real analytic function of $\theta$\ (for note that
$\cos\theta$ and $\sin\theta$\ are real analytic functions, and real analytic
functions are closed under composition). \ By Theorem \ref{aklmthm}, then, we
can take $\lambda_{1}\left(  \theta\right)  ,\ldots,\lambda_{N}\left(
\theta\right)  $\ to be real analytic\ functions as well.

By (iii), the acceptance probability $a\left(  \theta\right)  $ of $V$
(maximized over all witnesses) is equal to $\max_{i}\lambda_{i}\left(
\theta\right)  $. \ If $V$\ is a valid $\mathsf{QMA}_{1}$\ verifier, then we
must have $a\left(  0\right)  \leq1/2$, but $a\left(  \theta\right)  =1$\ for
all real $1\leq\theta\leq2$. \ Since $N$ is finite and the $\lambda_{i}\left(
\theta\right)  $'s are continuous, this implies that there exists an
$i\in\left[  N\right]  $ such that $\lambda_{i}\left(  0\right)  \leq1/2$, but
$\lambda_{i}\left(  \theta\right)  =1$\ for all $\theta$\ in some open
interval $\left(  x,y\right)  \subset\left[  1,2\right]  $. \ But this
contradicts the analyticity of $\lambda_{i}$ by Proposition \ref{analprop}.
\ Hence there must be a choice of\ $\theta$ such that $V$ does not solve the
problem correctly given $U=U\left(  \theta\right)  $\ as oracle.

We now simply diagonalize over all $n$ to achieve the desired oracle
separation. \ More formally,\ let $\mathcal{U}$ be a collection of quantum
oracles $U_{1},U_{2},\ldots$, such that $U_{n}=U\left(  \theta_{n}\right)
$\ rotates by the angle $\theta_{n}\in\left\{  0\right\}  \cup\left[
1,2\right]  $. \ Also, let $L$\ be a unary language such that $0^{n}\in L$\ if
and only if $\theta_{n}\neq0$. \ Then clearly $L\in\mathsf{BQP}^{\mathcal{U}}%
$, and hence $L\in\mathsf{QMA}^{\mathcal{U}}$, for all choices of
$\mathcal{U}$. \ On the other hand, we claim that $\mathcal{U}$\ can be chosen
so that $L\notin\mathsf{QMA}_{1}^{\mathcal{U}}$. \ To see this, let
$M_{1},M_{2},\ldots$\ be an enumeration of $\mathsf{QMA}_{1}$\ machines.
\ Then for each $i$, we simply choose an $n_{i}$ so large that $U_{n_{i}}$
cannot have been queried by machines $M_{1},\ldots,M_{i-1}$, and then set
$\theta_{n_{i}}$\ so that $M_{i}$\ fails on input $0^{n_{i}}$. \ (In other
words, either $\theta_{n_{i}}=0$\ and there exists a witness $\left\vert
\varphi\right\rangle $\ causing $M_{i}$\ to accept with probability greater
than $1/2$, or else $\theta_{n_{i}}\in\left[  1,2\right]  $\ and no witness
causes $M_{i}$\ to accept with probability $1$.) \ This is clearly possible by
the argument above.
\end{proof}

Notice that the proof of Theorem \ref{osep} breaks down if either $T$\ (the
number of queries to the unitary $U$) or $Q$ (the size of the witness) is
infinite. \ This is not an accident. \ If $T$ is infinite, then a quantum
algorithm can determine $\theta$\ exactly, with no need for a witness. \ If
$Q$ is infinite, then the witness $\left\vert \varphi\right\rangle $\ can
describe $\theta$\ to infinite precision, and verifying the description (with
perfect completeness) requires just a single query to $U$.

\section{Discussion\label{DISC}}

Perhaps the strangest aspect of Theorem \ref{osep} is its lack of dependence
on the polynomial running time of the $\mathsf{QMA}_{1}$\ machine. \ For
example, the same argument gives a quantum oracle $\mathcal{U}$\ such that
$\mathsf{BQP}^{\mathcal{U}}\not \subset \mathsf{QMA}_{1}\mathsf{EXP}%
^{\mathcal{U}}$, where $\mathsf{QMA}_{1}\mathsf{EXP}$\ is the exponential-time
version of $\mathsf{QMA}_{1}$, and even $\mathsf{BQP}^{\mathcal{U}%
}\not \subset \mathsf{QMA}_{1}\mathsf{EEEXP}^{\mathcal{U}}$. \ Indeed, just by
using Proposition \ref{analprop} about real analytic functions, without Lemma
\ref{vlem}\ or the theorem of Alekseevsky et al.\ \cite{aklm}, one can
construct a quantum oracle $\mathcal{U}$\ such that (for example)
$\mathsf{BQP}^{\mathcal{U}}\not \subset \mathsf{ZQEXP}^{\mathcal{U}}$, where
$\mathsf{ZQEXP}$\ is Zero-Error Quantum Exponential-Time.\footnote{If we just
want to separate $\mathsf{BQP}$\ from $\mathsf{ZQP}$ (Zero-Error Quantum
Polynomial-Time), this can be done with an ordinary classical oracle. \ Indeed
we can easily construct an oracle $A$ such that $\mathsf{BPP}^{A}%
\not \subset \mathsf{ZQP}^{A}$, by considering a problem where the answer is
YES if $A\left(  y\right)  =1$\ for most $y\in\left\{  0,1\right\}  ^{n}$, or
NO if $A\left(  y\right)  =0$\ for most $y\in\left\{  0,1\right\}  ^{n}$.
\ Such a problem is trivially in $\mathsf{BPP}^{A}$, but can be shown not to
be in $\mathsf{ZQP}^{A}$\ using the polynomial method of Beals et
al.\ \cite{bbcmw}. \ It would be nice if the same trick gave us a classical
oracle $A$ such that $\mathsf{BPP}^{A}\not \subset \mathsf{QMA}_{1}^{A}$, but
of course it does not, since the result of Zachos and F\"{u}rer \cite{zf}
(which is relativizing) implies that $\mathsf{BPP}^{A}\subseteq\mathsf{MA}%
^{A}=\mathsf{MA}_{1}^{A}\subseteq\mathsf{QMA}_{1}^{A}$ for all $A$.}

What makes this strange is that we know, by trivial relativizing arguments,
that $\mathsf{BQP}\subseteq\mathsf{EXP}\subseteq\mathsf{ZQEXP}$. \ Reflecting
on the apparent contradiction, one might suspect that the quantum oracle
separations are \textquotedblleft cheating\textquotedblright\ somehow. \ But
this is not the case; the correct resolution is simply that results like
$\mathsf{BQP}\subseteq\mathsf{ZQEXP}$, while classically relativizing, must be
quantumly non-relativizing! \ But how could that be?

If we carefully write out a proof that $\mathsf{BQP}\subseteq\mathsf{ZQEXP}$,
we see what the problem is. \ Since $\mathsf{ZQEXP}$\ is a zero-error class,
the \textquotedblleft obvious\textquotedblright\ proof will have to proceed
not by direct simulation of the $\mathsf{BQP}$\ machine, but by representing
the amplitudes of the $\mathsf{BQP}$\ machine in some explicit way. \ (In
other words, by mimicking the proofs of containments such as $\mathsf{BQP}%
\subseteq\mathsf{EXP}$ or $\mathsf{BQP}\subseteq\mathsf{PSPACE}$ \cite{bv}.)
\ But the technique of explicitly representing amplitudes, simple though it
seems, is already quantumly non-relativizing: it can break down if there is a
quantum oracle $\mathcal{U}$, some property of which must be decided without error!

Some readers might conclude from this that quantum oracles are
illegitimate;\ others, that\ the whole problem comes from the introduction of
one-sided-error quantum complexity classes like $\mathsf{QMA}_{1}$. \ Our own
view is that questions of \textquotedblleft complexity-theoretic
legitimacy\textquotedblright\ need to be decided on a case-by-case basis. \ In
the present case, the real substance of our result is that any proof of
$\mathsf{QMA}_{1}=\mathsf{QMA}$\ will need to involve explicit representation
of amplitudes (or something similar), rather than just black-box composition
of quantum circuits.

It remains a major challenge to find a quantumly non-relativizing technique
that both (i) goes beyond the known classically non-relativizing techniques
such as arithmetization, and (ii) fails to relativize even with
two-sided-error complexity classes.

\section{Extensions and Open Problems\label{EXT}}

By analogy to our quantum oracle separating $\mathsf{QMA}_{1}$\ from
$\mathsf{QMA}$,\ one might ask whether it is possible to construct a
\textquotedblleft randomized oracle\textquotedblright\ $R$ separating
$\mathsf{MA}_{1}$ from $\mathsf{MA}$.\footnote{It is also interesting to see
why a classical version of our argument does \textit{not} yield an ordinary
classical oracle $A$ such that $\mathsf{MA}_{1}^{A}\neq\mathsf{MA}^{A}$,
thereby contradicting the result of Zachos and F\"{u}rer \cite{zf} (which is
relativizing). \ The answer turns out to involve the fact that in the
classical case, Merlin can take advantage of the individual oracle bits,
rather than just the total amplitude for a `$1$' outcome. \ To put it another
way: in the classical case, there is no such thing as an oracle that is both
continuous and deterministic.} \ This would show that the proof of
$\mathsf{MA}_{1}=\mathsf{MA}$\ due to\ Zachos and F\"{u}rer \cite{zf} must
have been \textquotedblleft randomly non-relativizing.\textquotedblright%
\ \ Indeed such a randomized oracle separation is possible: simply have $R$
either output $0$ whenever it is queried (the NO case), or else output $0$ or
$1$ with equal probability (the YES\ case). \ It is obvious that these two
cases can be distinguished by a $\mathsf{BPP}^{R}$\ machine, using $O\left(
1\right)  $\ queries to $R$. \ On the other hand, because of the perfect
completeness requirement, the two cases \textit{cannot} be distinguished by an
$\mathsf{MA}_{1}^{R}$\ machine: having a witness in support of the YES case
clearly makes no difference.

However, this classical counterpart of our result really \textit{does} feel
like cheating! \ With the randomized oracle, perfect completeness is
unachievable for trivial information-theoretic reasons, even assuming an
infinitely long $\mathsf{MA}$\ witness. \ With the quantum oracle, by
contrast, perfect completeness \textit{would} be achievable, if there were
only some way to specify $\theta$\ to infinite precision using the quantum
witness $\left\vert \varphi\right\rangle $. \ This is of course what Theorem
\ref{osep} rules out.

The above discussion immediately suggests another question. \ In constructing
the quantum oracle $\mathcal{U}$, can we ensure that the angles $\theta_{n}$
are all rational numbers (or belong to some other dense countable set)?
\ Indeed, the proof of Theorem \ref{osep}\ can easily be modified to achieve
this. \ This is because of the following extension of Proposition
\ref{analprop}:

\begin{proposition}
Given a real analytic function $f:\mathbb{R}\rightarrow\mathbb{R}$, if there
exists an open set $\left(  x,y\right)  \subset\mathbb{R}$\ such that
$f\left(  z\right)  =1$\ for all rational points $z\in\left(  x,y\right)  $,
then $f\left(  z\right)  =1$\ identically.
\end{proposition}

However, there is an interesting difference between the real and rational
cases. \ In the case where the $\theta_{n}$'s are real, it is possible to
construct a \textit{single} quantum oracle $\mathcal{U}$\ such
that\ $\mathsf{BQP}^{\mathcal{U}}\not \subset \mathsf{QMA}_{1}\mathsf{TIME}%
\left(  f\left(  n\right)  \right)  ^{\mathcal{U}}$ for all functions $f$.
\ For example, choosing each $\theta_{n}$ to be $0$ with probability $1/2$, or
uniformly distributed in $\left[  1,2\right]  $\ with probability $1/2$, will
yield such a $\mathcal{U}$\ with probability $1$, by an argument due to
Bennett and Gill \cite{bg}. \ In the rational case, such a strong separation
is also achievable, but only by choosing the numerator and denominator of each
rational number $\theta_{n}$\ to grow faster than any computable function of
$n$. \ If we sidestep the issue of computability, say by giving the function
$f\left(  n\right)  $\ to the $\mathsf{QMA}_{1}\mathsf{TIME}\left(  f\left(
n\right)  \right)  $\ machine as advice, then it is not hard to show the following:

\begin{claim}
Given any quantum oracle $\mathcal{U}=\left\{  U_{n}\right\}  _{n\geq1}$\ with
rational angles $\left\{  \theta_{n}\right\}  _{n\geq1}$,\ there exists a
function $f$ such that $\mathsf{BQP}^{\mathcal{U}}$\ is simulable by a
zero-error quantum algorithm that makes $f\left(  n\right)  $\ queries to
$\mathcal{U}$.
\end{claim}

We end with two open problems. \ First, suppose the rotation angle $\theta
_{n}$\ cannot assume a continuum of values, but only a large finite set of
values $S_{n}$. \ Is it then the case that either $T$\ must scale like
$\left\vert S_{n}\right\vert ^{\Omega\left(  1\right)  }$\ or $Q$\ must scale
like $\Omega\left(  \log\left\vert S_{n}\right\vert \right)  $? \ What is the
optimal tradeoff between $T$ and $Q$? \ Are quantum witnesses (in this
setting) ever more powerful than classical witnesses of comparable size?

Second, can we prove a \textit{classical} oracle separation between
$\mathsf{QMA}_{1}$\ and $\mathsf{QMA}$?

\section{Acknowledgments}

I thank Andy Drucker and the anonymous reviewers for helpful comments, and
Greg Kuperberg and Dave Xiao for discussions of related problems several years ago.

\bibliographystyle{plain}
\bibliography{thesis}

\end{document}